\documentclass[conference]{IEEEtran}

\usepackage{calc}
\usepackage{cite}
\usepackage{soul}
\usepackage{tikz}
\usepackage[normalem]{ulem}
\usepackage{color}
\usepackage{cuted}
\usepackage{float}
\usepackage{amsthm}
\usepackage{lipsum}
\usepackage{amsmath}
\usepackage{amssymb}
\usepackage[small]{caption}
\usepackage{graphicx}
\usepackage{pgfplots}
\usepackage{stmaryrd}
\usepackage{thmtools}
\usepackage{tikzscale}
\usepackage{algorithm}
\usepackage{mathtools}
\usepackage{microtype}
\usepackage{subcaption}
\usepackage{algorithmic}
\usepackage{thm-restate}

\pgfplotsset{compat=newest}
\usetikzlibrary{arrows}
\usetikzlibrary{patterns}
\usetikzlibrary{plotmarks}
\usetikzlibrary{arrows.meta}
\usetikzlibrary{positioning}
\usetikzlibrary{external}
\usepgfplotslibrary{patchplots}

\makeatletter
\def\endthebibliography{%
	\def\@noitemerr{\@latex@warning{Empty `thebibliography' environment}}%
	\endlist
}
\makeatother

\allowdisplaybreaks

\newcommand{\abs}[1]{\left|#1\right|}

\theoremstyle{general} 
\theoremstyle{general} 
\theoremstyle{general} 
\theoremstyle{general} \newtheorem{proposition}{Proposition}
\theoremstyle{general} 
\theoremstyle{remark} 
\theoremstyle{remark} \newtheorem{remark}{Remark}

\newlength{\mygraphwidth}
\newlength{\mygraphheight}
\IEEEoverridecommandlockouts
\begin{document}
	\title{Precoder Design and Power Allocation for Downlink MIMO-NOMA via Simultaneous Triangularization}
	\author{\IEEEauthorblockN{Aravindh Krishnamoorthy\rlap{\textsuperscript{\IEEEauthorrefmark{1}\IEEEauthorrefmark{2}}},\,\,\, Meng Huang\rlap{\textsuperscript{\IEEEauthorrefmark{1}}}, and Robert Schober\rlap{\textsuperscript{\IEEEauthorrefmark{1}}}}\IEEEauthorblockA{\small \IEEEauthorrefmark{1}Friedrich-Alexander-Universit\"{a}t Erlangen-N\"{u}rnberg, \IEEEauthorrefmark{2}Fraunhofer Institute for Integrated Circuits (IIS)}}
	\maketitle
	
	\begin{abstract}
		In this paper, we consider the downlink precoder design for two-user power-domain multiple-input multiple-output (MIMO) non-orthogonal multiple access (NOMA) systems. The proposed precoding scheme is based on simultaneous triangularization and decomposes the MIMO-NOMA channels of the two users into multiple single-input single-output NOMA channels, assuming low-complexity self-interference cancellation at the users. In contrast to the precoding schemes based on simultaneous diagonalization (SD), the proposed scheme avoids inverting the MIMO channels of the users, thereby enhancing the ergodic rate performance. Furthermore, we develop a power allocation algorithm based on the convex-concave procedure, and exploit it to obtain the ergodic achievable rate region of the proposed MIMO-NOMA scheme. Our results illustrate that the proposed scheme outperforms baseline precoding schemes based on SD and orthogonal multiple access for a wide range of user rates and performs close to the dirty paper coding upper bound. The ergodic rate region can further be improved by utilizing a hybrid scheme based on time sharing between the proposed MIMO-NOMA scheme and point-to-point MIMO.
	\end{abstract}
	
	\section{Introduction}
	Multiple-input multiple-output (MIMO) non-orthogonal multiple access (NOMA) is a hybrid technology combining the advantages of the multiple spatial streams facilitated by MIMO with the non-orthogonal reception enabled by NOMA \cite{Saito2013}. MIMO-NOMA is one of the promising candidate technologies for the fifth generation (5G) and beyond communication systems. Furthermore, power-domain NOMA, which utilizes superposition coding at the transmitter and successive interference cancellation (SIC) based decoding at the receiver \cite{Saito2013}, is of interest owing to its simplicity and compatibility with the existing 4G networks.
	
	Several works have shown that power-domain MIMO-NOMA enables significantly higher data rates compared to MIMO orthogonal multiple access (OMA) \cite{Liu2016,Zeng2017}. However, a careful MIMO precoder design was found to be crucial for realizing the potential performance gains of MIMO-NOMA \cite{Ding2016}. To this end, multiple precoding schemes for MIMO-NOMA have been reported in literature \cite{Ding2016,Ali2017,Zeng2017a,Chen2016,Chen2017a,Ding2016b,Choi2016}. Furthermore, power allocation for MIMO-NOMA systems was investigated in \cite{Wang2019,Xiao2019}. Moreover, in order to reduce the decoding complexity at the users, precoder designs that simultaneously diagonalize the users' MIMO channels were reported in \cite{Chen2019} and \cite{Krishnamoorthy2019} and a simultaneously triangularizing uplink precoder design was proposed in \cite{Krishnamoorthy2019a}. The optimal power allocation for the precoder in \cite{Chen2019} was studied in \cite{Hanif2019}.
	
	The simultaneous diagonalization (SD) based precoding schemes in \cite{Chen2019} and \cite{Krishnamoorthy2019} enable low-complexity decoding at the users by decomposing the MIMO-NOMA channels of the users into multiple parallel single-input single-output (SISO)-NOMA channels. Furthermore, both schemes exploit the available null spaces of the MIMO channels of the users for enabling inter-user-interference free communication, thereby enhancing the ergodic rate performance. However, both precoding schemes achieve SD by inverting the MIMO channels of the users, which limits their performance.
	
	In this paper, we investigate the precoder design for a two-user power-domain MIMO-NOMA downlink system via simultaneous triangularization (ST). ST can overcome the limitations of SD as it avoids channel inversion. Nevertheless, ST can also decompose the MIMO-NOMA channels into SISO-NOMA channels, assuming low-complexity self-interference cancellation at the users, and takes advantage of the null spaces of the MIMO channels of the users for enabling inter-user-interference free communication.
	
	The main contributions of this paper are as follows.
	\begin{itemize}
		\item Exploiting the QR decomposition, we develop an ST MIMO-NOMA precoding scheme and a low-complexity decoding scheme, which decompose the downlink MIMO-NOMA channel into multiple parallel SISO-NOMA channels, assuming self-interference cancellation at the users.
		\item We develop a power allocation algorithm based on the convex-concave procedure (CCP) \cite{Yuille2003}, \cite{Lipp2016} and exploit it to obtain a lower bound on the ergodic achievable rate region of the proposed MIMO-NOMA scheme.
		\item We study the convergence of the proposed power allocation algorithm and the ergodic achievable rate region of the proposed MIMO-NOMA scheme based on computer simulations. Furthermore, we compare the ergodic achievable rate region of the proposed scheme with those of dirty paper coding (DPC) \cite{Vishwanath2003}, SD precoding \cite{Chen2019}, \cite{Krishnamoorthy2019}, and OMA.
	\end{itemize}

	The remainder of this paper is organized as follows. We describe the system model in Section \ref{sec:systemmodel}. In Section \ref{sec:proposed}, we present the proposed ST MIMO-NOMA precoding and decoding schemes. The proposed power allocation algorithm is provided in Section \ref{sec:poweralloc}, and simulation results are presented in Section \ref{sec:sim}. The paper is concluded in Section \ref{sec:con}.
	
	\emph{Notation:} Boldface capital letters $\boldsymbol{X}$ and lower case letters $\boldsymbol{x}$ denote matrices and vectors, respectively. $\boldsymbol{X}^\mathrm{T}$ and $\boldsymbol{X}^\mathrm{H}$ denote the transpose and Hermitian transpose of matrix $\boldsymbol{X}$, respectively. $\mathbb{C}^{m\times n}$ and $\mathbb{R}^{m\times n}$ denote the sets of all $m\times n$ matrices with complex-valued and real-valued entries, respectively. The $(i,j)$-th entry of matrix $\boldsymbol{X}$ is denoted by $[\boldsymbol{X}]_{ij}$ and the $i$-th entry of vector $\boldsymbol{x}$ is denoted by $[\boldsymbol{x}]_i.$ $\boldsymbol{I}_n$ denotes the $n\times n$ identity matrix, and $\boldsymbol{0}$ denotes the all zero matrix of appropriate dimension. The circularly symmetric complex Gaussian (CSCG) distribution with mean vector $\boldsymbol{\mu}$ and covariance matrix $\boldsymbol{\Sigma}$ is denoted by $\mathcal{CN}(\boldsymbol{\mu},\boldsymbol{\Sigma})$; $\sim$ stands for ``distributed as.''
	\section{System Model}
	\label{sec:systemmodel}
	We consider a two-user\footnote{We restrict the number of paired users to two for problem tractability and to limit the overall decoding complexity at the receivers as pairing $K > 1$ users necessitates $K(K-1)/2$ successive interference cancellation stages. For $K>2$ users, a hybrid approach, such as in \cite[Section V-B]{Chen2019}, can be employed where users are divided into groups of two users and each group is allocated orthogonal resources. Within each two-user group, the proposed MIMO-NOMA scheme can be applied.} downlink transmission, where the BS is equipped with $N$ antennas, and the users employ $M_1$ and $M_2$ antennas, respectively. Furthermore, we assume that the first user is located farther away from the BS and experiences a higher path loss compared to the second user.

	The MIMO channel matrix between the BS and user $k$ is modeled as
	\begin{equation}
	\frac{1}{\sqrt{\mathstrut \Pi_k}} \boldsymbol{H}_k, \label{eqn:plm}
	\end{equation}
	where matrix $\boldsymbol{H}_k \in \mathbb{C}^{M_k\times N}$ captures the small-scale fading effects, and its elements $[\boldsymbol{H}_k]_{ij} \sim \mathcal{CN}(0,1), i=1,\dots,M_k, j=1,\dots,N, k=1,2,$ are statistically independent for all $i,j,k.$ The scalar $\Pi_k > 0, k=1,2,$ models the path loss between the BS and user $k,$ where $\Pi_1 > \Pi_2.$ Perfect knowledge of both MIMO channel matrices, $\boldsymbol{H}_1$ and $\boldsymbol{H}_2,$ is assumed at the BS whereas perfect knowledge of their respective MIMO channel matrices is assumed at the users.
	
	Let $L = \mathrm{min}\left\{M_1+M_2, N\right\}$ denote the symbol vector length, and let $\boldsymbol{s}_1 = [s_{1,1},\dots,s_{1,L}]^\mathrm{T} \in \mathbb{C}^{L\times 1}$ and $\boldsymbol{s}_2 = [s_{2,1},\dots,s_{2,L}]^\mathrm{T} \in \mathbb{C}^{L\times 1}$ denote the symbol vectors intended for the first and the second users, respectively. We assume that the $s_{k,l} \sim \mathcal{CN}(0,1), k=1,2, l=1,\dots,L,$ are statistically independent for all $k,l.$ We construct the MIMO-NOMA symbol vector $\boldsymbol{s} = [s_1, \dots, s_L]^\mathrm{T}$ as follows
	\begin{align}
	\boldsymbol{s} = \scalebox{0.85}{\mbox{\ensuremath{\displaystyle \mathrm{diag}\left(\sqrt{p_{1,1}},\dots,\sqrt{p_{1,L}}\right)\boldsymbol{s}_1 + \mathrm{diag}\left(\sqrt{p_{2,1}},\dots,\sqrt{p_{2,L}}\right)\boldsymbol{s}_2}}}, \label{eqn:s}
	\end{align}
	where $p_{k,l} \geq 0, k=1,2, l=1,\dots,L,$ is the transmit power allocated to the $l$-th symbol of user $k.$ The MIMO-NOMA symbol vector is precoded using a linear precoder matrix $\boldsymbol{P} \in \mathbb{C}^{N\times L}$ such that 
	\begin{align}
	\scalebox{0.95}{\mbox{\ensuremath{\displaystyle {\mathrm{tr}\left(\boldsymbol{P} \mathrm{diag}\left(p_{1,1} + p_{2,1}, \dots, p_{1,L} + p_{2,L}\right) \boldsymbol{P}^H\right)} \leq P_T}}}, \label{eqn:eppleq1}
	\end{align}
	where $P_T$ denotes the maximum available transmit power, resulting in transmit signal $\boldsymbol{x} = \boldsymbol{P} \boldsymbol{s}.$ 
	
	At user $k,$ $k=1,2,$ the received signal, $\hat{\boldsymbol{y}}_k \in \mathbb{C}^{M_k\times 1},$ is given by
	\begin{align}
	\hat{\boldsymbol{y}}_k &= \frac{1}{\sqrt{\Pi_k}}\boldsymbol{H}_k \boldsymbol{x} + \hat{\boldsymbol{z}}_k = \frac{1}{\sqrt{\Pi_k}} \boldsymbol{H}_k \boldsymbol{P} \boldsymbol{s} + \hat{\boldsymbol{z}}_k,
	\end{align} 
	where $\hat{\boldsymbol{z}}_k \sim \mathcal{CN}(\boldsymbol{0},\sigma^2 \boldsymbol{I}_{M_k})$ denotes the additive white Gaussian noise (AWGN) vector at user $k.$ Furthermore, at user $k,$ signal $\hat{\boldsymbol{y}}_k$ is processed by a unitary detection matrix $\boldsymbol{Q}_k \in \mathbb{C}^{M_k\times M_k}$ leading to
	\begin{equation}
	\boldsymbol{y}_k = \boldsymbol{Q}_k \hat{\boldsymbol{y}}_k = \frac{1}{\sqrt{\Pi_k}} \boldsymbol{Q}_k\boldsymbol{H}_k \boldsymbol{P} \boldsymbol{s} + \boldsymbol{z}_k, \label{eqn:yk}
	\end{equation}
	where $\boldsymbol{z}_k = \boldsymbol{Q}_k\hat{\boldsymbol{z}}_k \sim \mathcal{CN}(\boldsymbol{0},\sigma^2 \boldsymbol{I}_{M_k}).$ $\boldsymbol{y}_k$ is subsequently used for detection.
	\section{Proposed ST Precoding Scheme}
	\label{sec:proposed}
	In this section, we present the proposed ST precoding scheme and derive expressions for the resulting achievable rates of the users.
	
	\subsection{Simultaneous Triangularization}
	Let $\bar{M}_1 = \mathrm{min}\left\{M_1,L-M_2\right\}, \bar{M}_2 = \mathrm{min}\left\{M_2,L-M_1\right\},$ and $M = N-\bar{M}_1-\bar{M}_2.$ ST of $\boldsymbol{H}_1$ and $\boldsymbol{H}_2$ is compactly stated in the following proposition.
	
	\begin{proposition}
		\label{prop:st}
		Let $\boldsymbol{H}_1$ and $\boldsymbol{H}_2$ be as defined in Section \ref{sec:systemmodel}. Then, there exist unitary matrices $\boldsymbol{Q}_1 \in \mathbb{C}^{M_1\times M_1}, \boldsymbol{Q}_2 \in \mathbb{C}^{M_2\times M_2},$ and a full matrix $\boldsymbol{X} \in \mathbb{C}^{N\times L}$ such that
		\begin{align}
		\boldsymbol{Q}_1\boldsymbol{H}_1\boldsymbol{X} &= \begin{bmatrix}\boldsymbol{R}_1 & \boldsymbol{0}\end{bmatrix}, \label{eqn:st1}\\
		\boldsymbol{Q}_2\boldsymbol{H}_2\boldsymbol{X} &= \begin{bmatrix}\boldsymbol{R}_2' & \boldsymbol{0} & \boldsymbol{R}_2''\end{bmatrix}, \label{eqn:st2}
		\end{align}
		where $\boldsymbol{R}_1 \in \mathbb{C}^{M_1\times (M+\bar{M}_1)}, \boldsymbol{R}_2' \in \mathbb{C}^{M_2\times M},$ and $\boldsymbol{R}_2'' \in \mathbb{C}^{M_2\times \bar{M}_2}.$ Furthermore, $\boldsymbol{R}_1$ and $\boldsymbol{R}_2 = \begin{bmatrix}\boldsymbol{R}_2' & \boldsymbol{R}_2''\end{bmatrix} \in \mathbb{C}^{M_2\times (M+\bar{M}_2)}$ are upper-triangular matrices with real-valued entries on their main diagonals.
	\end{proposition}
	\begin{proof}
		Let  $\bar{\boldsymbol{H}}_1 \in \mathbb{C}^{N\times \bar{M}_1}$ and $\bar{\boldsymbol{H}}_2 \in \mathbb{C}^{N\times \bar{M}_2}$ be matrices that contain a basis for the null space of $\boldsymbol{H}_1$ and $\boldsymbol{H}_2,$ respectively. Let $\boldsymbol{K} \in \mathbb{C}^{N\times M}$ denote the matrix containing a basis for the null space of $\begin{bmatrix} \bar{\boldsymbol{H}}_1^\mathrm{H} & \bar{\boldsymbol{H}}_2^\mathrm{H}\end{bmatrix}.$ When the null spaces of $\boldsymbol{H}_1$ and $\boldsymbol{H}_2$ are trivial, i.e., when $M_1, M_2 \geq N,$ then $\boldsymbol{K} = \boldsymbol{I}_N.$ Let, by QR decomposition,
		\begin{align}
			\boldsymbol{\mathcal{Q}}_1 \boldsymbol{R}_1 &= \boldsymbol{H}_1 \begin{bmatrix}\boldsymbol{K} & \bar{\boldsymbol{H}}_2\end{bmatrix}, \label{eqn:qr1}\\
			\boldsymbol{\mathcal{Q}}_2 \boldsymbol{R}_2 &= \boldsymbol{H}_2 \begin{bmatrix}\boldsymbol{K} & \bar{\boldsymbol{H}}_1\end{bmatrix}. \label{eqn:qr2}
		\end{align}
		 Then, (\ref{eqn:st1}) and (\ref{eqn:st2}) are satisfied by setting
		\begin{align}
			\boldsymbol{X} = \begin{bmatrix}\boldsymbol{K} & \bar{\boldsymbol{H}}_2 & \bar{\boldsymbol{H}}_1\end{bmatrix}, \label{eqn:x}
		\end{align}
		 and choosing $\boldsymbol{Q}_1 = \boldsymbol{\mathcal{Q}}_1^\mathrm{H}$ and $\boldsymbol{Q}_2 = \boldsymbol{\mathcal{Q}}_2^\mathrm{H}$ from (\ref{eqn:qr1}) and (\ref{eqn:qr2}) above, to obtain
		\begin{align}
		\scalebox{0.85}{\mbox{\ensuremath{\displaystyle \boldsymbol{Q}_1\boldsymbol{H}_1\boldsymbol{X}}}} &= \scalebox{0.85}{\mbox{\ensuremath{\displaystyle \begin{bmatrix}\underbrace{\boldsymbol{\mathcal{Q}}_1^\mathrm{H} \boldsymbol{H}_1 \begin{bmatrix}\boldsymbol{K} & \bar{\boldsymbol{H}}_2\end{bmatrix}}_{\boldsymbol{R}_1} &  \underbrace{\boldsymbol{\mathcal{Q}}_1^\mathrm{H} \boldsymbol{H}_1 \bar{\boldsymbol{H}}_1}_{\boldsymbol{0}}\end{bmatrix},}}} \label{eqn:q1h1x} \\ 
			\scalebox{0.85}{\mbox{\ensuremath{\displaystyle \boldsymbol{Q}_2\boldsymbol{H}_2\boldsymbol{X}}}} &\scalebox{0.85}{\mbox{\ensuremath{\displaystyle \overset{(a)}{=} \begin{bmatrix} \underbrace{\boldsymbol{\mathcal{Q}}_2^\mathrm{H} \boldsymbol{H}_2 \boldsymbol{K}}_{\boldsymbol{R}_2'} & \underbrace{\boldsymbol{\mathcal{Q}}_2^\mathrm{H} \boldsymbol{H}_2 \bar{\boldsymbol{H}}_2}_{\boldsymbol{0}} &  \underbrace{\boldsymbol{\mathcal{Q}}_2^\mathrm{H} \boldsymbol{H}_2 \bar{\boldsymbol{H}}_1}_{\boldsymbol{R}_2''}\end{bmatrix},}}}\label{eqn:q2h2x}
		\end{align}
		where (a) holds because the QR decomposition in (\ref{eqn:qr2}) is unaffected by the zero columns introduced in the middle.
	\end{proof}	
	
	\subsection{Proposed Precoding Scheme}
	\label{sec:prec}
	Based on Proposition \ref{prop:st}, the precoder matrix can be chosen as $\boldsymbol{P} = \boldsymbol{X},$ and the detection matrices of users 1 and 2 can be chosen directly as $\boldsymbol{Q}_1$ and $\boldsymbol{Q}_2$ for users 1 and 2, respectively\footnote{Note that although the proposed scheme utilizes QR decomposition based detection matrices, other reception schemes such as zero forcing can also be utilized for the proposed precoder in (\ref{eqn:x}).}. Hence, the received signal at the users, based on (\ref{eqn:yk}), can be simplified to
	\begin{align}
		\tilde{\boldsymbol{y}}_k &= \frac{1}{\sqrt{\Pi_k}} \boldsymbol{R}_k \tilde{\boldsymbol{s}}_k + \tilde{\boldsymbol{z}}_k \label{eqn:ykr}
	\end{align}
	where, based on Proposition \ref{prop:st}, symbol vectors $\tilde{\boldsymbol{s}}_k \in \mathbb{C}^{(M+\bar{M}_k)\times 1}$ are defined as $[\tilde{\boldsymbol{s}}_k]_l = s_l,$ $l=1,\dots,M,$ $[\tilde{\boldsymbol{s}}_1]_{l+M} = s_{l+M},$ $l=1,\dots,\bar{M}_1,$ and $[\tilde{\boldsymbol{s}}_2]_{l+M} = s_{l+M+\bar{M}_1},$ $l=1,\dots,\bar{M}_2,$ and $\tilde{\boldsymbol{z}}_k$ contains the corresponding elements of $\boldsymbol{z}_k.$
	
	As symbols $s_{l}, l=M+1,\dots,M+\bar{M}_1,$ are only transmitted to user 1, $p_{2,l} = 0$ for $l=M+1,\dots,M+\bar{M}_1.$ Similarly, $p_{1,l} = 0$ for $l=M+\bar{M}_1+1,\dots,L.$ Power allocation coefficients $p_{k,l}, k=1,2, l=1,\dots,M,$ can be used to adjust the rates of users 1 and 2. Lastly, the $M_k\times(M+\bar{M}_k)$ matrices $\boldsymbol{R}_k, k=1,2,$ are of the form		
	\begin{equation}
	\scalebox{0.9}{\mbox{\ensuremath{\displaystyle \begin{bmatrix}
		\rho^{(k)}_{1,1} &  \dots   & \rho^{(k)}_{1,M+\bar{M}_k}   \\
		0		 & \ddots &  \vdots                       \\
		\vdots        &   0      & \rho^{(k)}_{M+\bar{M}_k,M+\bar{M}_k} \\
		\boldsymbol{0}    & \boldsymbol{0} &  \boldsymbol{0}\\
	\end{bmatrix}}}}. \label{eqn:Rk}
	\end{equation}

	\subsection{Decoding Scheme}
	\label{sec:dec}
	Let $\hat{s}_{k,l},$ $k=1,2,l=1,\dots,L,$ denote the detected symbols corresponding to transmitted symbols $s_{k,l}.$ As the rates of $s_{k,l},$ $k=1,2,l=1,\dots,L,$ are chosen such that they lie within the achievable rate region, perfect decoding, i.e., $\hat{s}_{k,l} = s_{k,l},$ is assumed in the following.
	
	From the upper triangular structure of $\boldsymbol{R}_k$ in (\ref{eqn:Rk}), and based on (\ref{eqn:st1}) and (\ref{eqn:st2}), we note that the symbols corresponding to the last $\bar{M}_k$ columns of $\boldsymbol{R}_k, k=1,2,$ which are only transmitted to user $k,$ contain no inter-user-interference and are therefore decoded directly, in reverse order. Next, the symbols corresponding to the first $M$ columns of $\boldsymbol{R}_k, k=1,2,$ which are transmitted to both users, are decoded as in SISO-NOMA \cite{Saito2013}, also in reverse order. For each symbol, the self-interference of the previously decoded symbols is eliminated. The decoding process is described in detail below.
	
	The first user decodes the symbols as follows.
	\begin{enumerate}
		\item If $\bar{M}_1 > 0,$ symbols $s_{1,l},$ $l=M+1,\dots,M+\bar{M}_1,$ are decoded, in reverse order, starting from the self-interference free element $[\tilde{\boldsymbol{y}}_1]_{M+\bar{M}_1}$ given in (\ref{eqn:ykr}). For each subsequent symbol, the self-interference from the previously decoded symbols is eliminated, resulting in the self-interference free signal
		\begin{align}
			\scalebox{0.75}{\mbox{\ensuremath{\displaystyle [\hat{\boldsymbol{y}}_1]_{l} = [\tilde{\boldsymbol{y}}_1]_{l} - \sum_{l'=l+1}^{M+\bar{M}_1} \sqrt{\frac{p_{1,l'}}{\Pi_1}}\rho^{(1)}_{l,l'}\hat{s}_{1,l'} = \sqrt{\frac{p_{1,l}}{\Pi_1}}\rho^{(1)}_{l,l} s_{1,l} + [\tilde{\boldsymbol{z}}_k]_l,}}} \label{eqn:dec12}
		\end{align}
		for $l=M+1,\dots,M+\bar{M}_1$ (in reverse order), which is decoded.
		
		\item Next, if $M > 0,$ symbols $s_{1,l},$ $l=1,\dots,M,$ which contain inter-user-interference from $s_{2,l},l=1,\dots,M,$ are decoded directly as in SISO-NOMA, in reverse order. Then, as above, for each symbol, the self-interference from the previously decoded symbols is eliminated, resulting in the self-interference free signal
		\begin{align}
			\scalebox{0.825}{\mbox{\ensuremath{\displaystyle [\hat{\boldsymbol{y}}_1]_{l}}}} &= \scalebox{0.825}{\mbox{\ensuremath{\displaystyle [\tilde{\boldsymbol{y}}_1]_{l} - \sum_{l'=l+1}^{M} \sqrt{\frac{p_{1,l'}}{\Pi_1}}\rho^{(1)}_{l,l'}\hat{s}_{1,l'} - \sum_{l'=M+1}^{M+\bar{M}_1} \sqrt{\frac{p_{1,l'}}{\Pi_1}}\rho^{(1)}_{l,l'}\hat{s}_{1,l'}}}} \nonumber\\
			&= \scalebox{0.9}{\mbox{\ensuremath{\displaystyle \sqrt{\frac{p_{1,l}}{\Pi_1}}\rho^{(1)}_{l,l} s_{1,l} + \sum_{l'=l}^{M} \sqrt{\frac{p_{2,l'}}{\Pi_1}}\rho^{(1)}_{l,l'}s_{2,l'} + [\tilde{\boldsymbol{z}}_k]_l}}}, \label{eqn:dec11}
		\end{align}
		for $l=1,\dots,M$ (in reverse order), which is decoded. Note that the residual inter-user-interference from symbols $s_{2,l},l=1,\dots,M,$ cannot be eliminated and is treated as noise.
		
	\end{enumerate}
	Analogously, at the second user decodes the symbols as follows.
		\begin{enumerate}
		\item If $\bar{M}_2 > 0,$ symbols $s_{2,l},$ $l=M+\bar{M}_1+1,\dots,M+\bar{M}_1+\bar{M}_2,$ are decoded, in reverse order, starting from the self-interference free element $[\tilde{\boldsymbol{y}}_2]_{M+\bar{M}_2}$ given in (\ref{eqn:ykr}). For each symbol, the self-interference from the previously decoded symbols is eliminated, resulting in the self-interference free signal
		\begin{align}
			\scalebox{0.85}{\mbox{\ensuremath{\displaystyle [\hat{\boldsymbol{y}}_2]_{l-\bar{M}_1}}}} &= \scalebox{0.85}{\mbox{\ensuremath{\displaystyle [\tilde{\boldsymbol{y}}_2]_{l-\bar{M}_1} - \sum_{l'=l+1}^{M+\bar{M}_1+\bar{M}_2} \sqrt{\frac{p_{2,l'}}{\Pi_2}}\rho^{(2)}_{l-\bar{M}_1,l'-\bar{M}_1}\hat{s}_{2,l'}}}} \nonumber\\
			&= \scalebox{0.85}{\mbox{\ensuremath{\displaystyle \sqrt{\frac{p_{2,l}}{\Pi_2}}\rho^{(2)}_{l-\bar{M}_1,l-\bar{M}_1} s_{2,l} + [\tilde{\boldsymbol{z}}_k]_{l-\bar{M}_1},}}} \label{eqn:dec22}
		\end{align}
		for $l=M+\bar{M}_1+1,\dots,M+\bar{M}_1+\bar{M}_2$ (in reverse order), which is decoded.
		
		\item Next, if $M > 0,$ symbols $s_{k,l},$ $k=1,2,l=1,\dots,M,$ are decoded as in SISO-NOMA \cite{Saito2013} where the first user's symbols are decoded directly and the second user's symbols are decoded after SIC, in reverse order. In this case also, for each symbol, the self-interference from the previously decoded symbols is eliminated, resulting in the self-interference free signal
		\begin{align}
			\scalebox{0.75}{\mbox{\ensuremath{\displaystyle [\hat{\boldsymbol{y}}_2]_{l}}}} &= \scalebox{0.75}{\mbox{\ensuremath{\displaystyle [\tilde{\boldsymbol{y}}_2]_{l} - \sum_{l'=l+1}^{M} \frac{1}{\sqrt{\Pi_1}}\rho^{(2)}_{l,l'}(\sqrt{p_{1,l'}}\hat{s}_{1,l'}+\sqrt{p_{2,l'}}\hat{s}_{2,l'})}}} \nonumber\\
				&\qquad\scalebox{0.75}{\mbox{\ensuremath{\displaystyle {}- \sum_{l'=M+\bar{M}_1+1}^{M+\bar{M}_1+\bar{M}_2} \sqrt{\frac{p_{2,l'}}{\Pi_2}}\rho^{(2)}_{l-\bar{M}_1,l'-\bar{M}_1}\hat{s}_{2,l'}}}} \nonumber\\
		&= \scalebox{0.8}{\mbox{\ensuremath{\displaystyle \rho^{(2)}_{l,l} \left(\sqrt{\frac{p_{1,l}}{\Pi_2}}s_{1,l} + \sqrt{\frac{p_{2,l}}{\Pi_2}}s_{2,l}\right) + [\tilde{\boldsymbol{z}}_k]_l,}}} \label{eqn:dec21}
		\end{align}
		for $l=1,\dots,M$ (in reverse order), which is decoded as in SISO-NOMA.
	\end{enumerate}

	As seen from the decoding scheme above, the MIMO-NOMA channels of the users are decomposed to scalar channels. Furthermore, the second user is always the SIC user, i.e., no SIC capability is necessary at the first user.

	\begin{remark}
	From (\ref{eqn:dec12}) and (\ref{eqn:dec22}), we observe that the proposed MIMO-NOMA precoding scheme exploits the available null spaces of the MIMO channel matrices of the users for transmitting $\bar{M}_1$ and $\bar{M}_2$ symbols inter-user-interference free to users 1 and 2, respectively.
	\end{remark}

	\subsection{Computational Complexity}
	For an $m\times n$ matrix $\boldsymbol{X}$, a matrix $\bar{\boldsymbol{X}}$ containing a basis for the null space can be computed using the QR decomposition, which entails a complexity of $\mathcal{O}\mkern-\medmuskip\left(2 n m^2\right)$ \cite[Alg. 5.2.5]{Golub2012}. Hence, constructing the precoder matrix in (\ref{eqn:x}) entails a total complexity of $\mathcal{O}\mkern-\medmuskip\left(2 N M_1^2 + 2 N M_2^2 + 2 N (N-M)^2\right)$ at the BS for computing $\bar{\boldsymbol{H}}_1, \bar{\boldsymbol{H}}_2,$ and $\boldsymbol{K}$ via the QR decomposition. Next, at user $k,$ $k=1,2,$ detection matrix $\boldsymbol{Q}_k,$ which is also computed using the QR decomposition, and self-interference cancellation entail complexities of $\mathcal{O}\mkern-\medmuskip\left(2 M_k^3\right)$ and $\mathcal{O}\mkern-\medmuskip\left(M_k^2\right),$ respectively. Therefore, the proposed scheme entails an overall worst case complexity of $\mathcal{O}\mkern-\medmuskip\left(N^3\right),$ which is identical to that of the SD schemes in \cite{Chen2019} and \cite{Krishnamoorthy2019}.

	\subsection{Achievable Rates}
	\label{sec:rates}
	At the first user, based on (\ref{eqn:dec11}), the achievable rate after self-interference cancellation, is given by
	\begin{align}
	\scalebox{0.85}{\mbox{\ensuremath{\displaystyle R_{1,l}^{(1)} = \log_2\left(1 + \frac{\frac{1}{\Pi_1}p_{1,l}\abs{\rho^{(1)}_{l,l}}^2}{\sigma^2 + \frac{1}{\Pi_1}\sum_{l'=l}^{M}p_{2,l'}\abs{\rho^{(1)}_{l,l'}}^2}\right),}}} \label{eqn:r1lm}
	\end{align}
	for $l = 1,\dots,M,$ and, based on (\ref{eqn:dec12}), by
	\begin{align}
	\scalebox{0.85}{\mbox{\ensuremath{\displaystyle R_{1,l}^{(1)} = \log_2\left(1 + \frac{p_{1,l}\abs{\rho^{(1)}_{l,l}}^2}{\Pi_1\sigma^2}\right),}}}
	\end{align}
	for $l = M+1,\dots,M+\bar{M}_1$. Furthermore, as mentioned in Section \ref{sec:prec}, $R_{1,l}^{(1)}=0$ for $l=M+\bar{M}_1+1,\dots,L.$
	
	Similarly, at the second user, the achievable rate for $s_{1,l}, l=1,\dots,M,$ after self-interference cancellation, based on (\ref{eqn:dec21}), is given by
	\begin{align}
		\scalebox{0.85}{\mbox{\ensuremath{\displaystyle R_{1,l}^{(2)} = \log_2\left(1 + \frac{\frac{1}{\Pi_2}p_{1,l}\abs{\rho^{(2)}_{l,l}}^2}{\sigma^2 + \frac{1}{\Pi_2}p_{2,l}\abs{\rho^{(2)}_{l,l}}^2}\right).}}} \label{eqn:r1lm2}
	\end{align}
	Note that the difference to (\ref{eqn:r1lm}) arises because the inter-user-interference from symbols $s_{2,l}, l=1,\dots,M,$ can be eliminated at the second user as both users' symbols are decoded per SISO-NOMA.
	
	In order to ensure that symbols $s_{1,l}, l=1,\dots,M,$ can be decoded at both users, the rates are chosen as the instantaneous minimum rate. Hence,
	\begin{align}
		R_{1,l} = \mathrm{min}\left\{R_{1,l}^{(1)}, R_{1,l}^{(2)}\right\}, \label{eqn:r1l}
	\end{align}
	for $l=1,\dots,M,$ and $R_{1,l} = R_{1,l}^{(1)}$ for $l=M+1,\dots,L.$
	
	Lastly, at the second user, based on (\ref{eqn:dec21}), for $l=1,\dots,M,$ the achievable rate after SIC and self-interference cancellation is given by
	\begin{align}
		\scalebox{0.85}{\mbox{\ensuremath{\displaystyle R_{2,l} = \log_2\left(1 + \frac{p_{2,l}\abs{\rho^{(2)}_{l,l}}^2}{\Pi_2\sigma^2}\right),}}}
	\end{align}
	and, based on (\ref{eqn:dec22}), for $l=M+\bar{M}_1+1,\dots,L,$ the achievable rate after self-interference cancellation is given by
	\begin{align}
		\scalebox{0.85}{\mbox{\ensuremath{\displaystyle R_{2,l} = \log_2\left(1 + \frac{p_{2,l}\abs{\rho^{(2)}_{l-\bar{M}_1,l-\bar{M}_1}}^2}{\Pi_2\sigma^2}\right).}}}
	\end{align}
	Furthermore, as mentioned in Section \ref{sec:prec}, $R_{2,l}=0$ for $l=M+1,\dots,M+\bar{M}_1.$
		
	\section{Power Allocation}
	\label{sec:poweralloc}
	In this section, we present a low-complexity algorithm for power allocation.
		
	We begin by noting that the columns of matrix $\boldsymbol{X}$ in (\ref{eqn:x}) are unit norm, and hence, the condition in (\ref{eqn:eppleq1}) simplifies to
	\begin{align}
		\sum_{l = 1}^{L} (p_{1,l} + p_{2,l}) \leq P_T. \label{eqn:npp}
	\end{align}
	
	Next, the weighted sum rate of the users can be expressed as
	\begin{align}
	R = \sum_{l=1}^{L} (\mu R_{1,l} + (1-\mu) R_{2,l}),
	\label{eqn:r}
	\end{align}
	where $\mu \in [0,1]$ is a fixed weight which can be chosen to adjust the achievable rates of user 1 and 2 during power allocation \cite[Sec. 4]{WangGiannakis2011}. In order to maximize the weighted sum rate (or equivalently to minimize the negative weighted sum rate) we formulate the following optimization problem:
	\begin{align}
		\text{P1: } \quad \min_{p_{k,l} \geq 0\;\forall\;k,l} -R \qquad \text{s.t.} \quad (\ref{eqn:npp}).
	\end{align}
	The rate region of the proposed ST MIMO-NOMA scheme can be obtained by solving P1 for different values of $\mu.$
	
	Optimization problem P1 is non-convex in the optimization variables $p_{k,l}, k=1,2, l=1,\dots,M,$ due to their coupling in (\ref{eqn:r1lm}) and (\ref{eqn:r1lm2}). Hence, computationally efficient convex optimization methods cannot be applied to obtain the global optimum. Instead, in the following, we reformulate problem P1 as a difference of convex (DC) programming problem which can be solved using the efficient CCP \cite{Yuille2003}, \cite{Lipp2016} to obtain a suboptimal solution.
	
	\subsection{Convex-Concave Procedure}
	We note, based on (\ref{eqn:r1lm}) and (\ref{eqn:r1lm2}), that the non-convexity of P1 is due to $-R_{1,l},$ $l=1,\dots,M.$ Hence, first, we rewrite (\ref{eqn:r1lm}) as
	\begin{align}
	\scalebox{0.85}{\mbox{\ensuremath{\displaystyle R_{1,l}^{(1)}}}} &\scalebox{0.85}{\mbox{\ensuremath{\displaystyle = \underbrace{\log_2\left(\sigma^2 + \frac{1}{\Pi_1}\sum_{l'=l}^{M}p_{2,l'}\abs{\rho^{(1)}_{l,l'}}^2 + \frac{1}{\Pi_1}p_{1,l}\abs{\rho^{(1)}_{l,l}}^2\right)}_{R_{1,l}^{(1,1)}}}}} \nonumber\\
	&\qquad\scalebox{0.85}{\mbox{\ensuremath{\displaystyle {} - \underbrace{\log_2\left(\sigma^2 + \frac{1}{\Pi_1}\sum_{l'=l}^{M}p_{2,l'}\abs{\rho^{(1)}_{l,l'}}^2\right)}_{R_{1,l}^{(1,2)}}}}},
	\end{align}
	for $l=1,\dots,M,$ where $R_{1,l}^{(1,1)}$ and $R_{1,l}^{(1,2)}$ are both jointly concave in the optimization variables $p_{k,l}, k=1,2, l=1,\dots,M.$ Analogously, (\ref{eqn:r1lm2}) can be rewritten as $R_{1,l}^{(2)} =  R_{1,l}^{(2,1)} - R_{1,l}^{(2,2)},$ where, $R_{1,l}^{(2,1)}$ and $R_{1,l}^{(2,2)}$ are similarly defined as $R_{1,l}^{(1,1)}$ and $R_{1,l}^{(1,2)}.$ Hence, (\ref{eqn:r1l}) can be simplified to
	\begin{align}
		\scalebox{0.825}{\mbox{\ensuremath{\displaystyle R_{1,l}}}} &= \scalebox{0.825}{\mbox{\ensuremath{\displaystyle \mathrm{min}\left\{R_{1,l}^{(1,1)} {-} R_{1,l}^{(1,2)}, R_{1,l}^{(2,1)} {-} R_{1,l}^{(2,2)}\right\}}}} \nonumber\\
		&\overset{(a)}{=} \scalebox{0.825}{\mbox{\ensuremath{\displaystyle \mathrm{min}\left\{R_{1,l}^{(1,1)} {+} R_{1,l}^{(2,2)}, R_{1,l}^{(2,1)} {+} R_{1,l}^{(1,2)}\right\} - (R_{1,l}^{(1,2)} {+} R_{1,l}^{(2,2)})}}}, \label{eqn:dc}
	\end{align}
	where (a) is obtained using the relation
	\begin{align}
		\mathrm{min}\left\{A{-}B,C{-}D\right\} = \mathrm{min}\left\{A{+}D,C{+}B\right\} - (B{+}D).
	\end{align}
	As the minimum of two concave functions is concave and the sum of two concave functions is concave \cite{Boyd2004}, $-R_{1,l}$ in (\ref{eqn:dc}) is in DC form. Hence, P1 is a DC optimization problem.
	
	Next, in accordance with CCP, we construct a linear underestimator for $-R_{1,l},l=1,\dots,M,$ based on a first-order approximation, around a feasible point $\boldsymbol{q} = \{q_1, \dots, q_M\}$ as
	\begin{align}
		\scalebox{0.825}{\mbox{\ensuremath{\displaystyle \tilde{R}_{1,l}(\boldsymbol{q})}}} &= \scalebox{0.825}{\mbox{\ensuremath{\displaystyle \mathrm{min}\left\{R_{1,l}^{(1,1)} + R_{1,l}^{(2,2)}, R_{1,l}^{(2,1)} + R_{1,l}^{(1,2)}\right\} + \left.\left(R_{1,l}^{(1,2)} + R_{1,l}^{(2,2)}\right)\right|_{\boldsymbol{q}}}}} \nonumber\\
		&\qquad\scalebox{0.825}{\mbox{\ensuremath{\displaystyle -\underbrace{\left.\frac{\partial}{\partial p_{2,l}} \left(R_{1,l}^{(1,2)} + R_{1,l}^{(2,2)}\right)\right|_{\boldsymbol{q}}\left(p_{2,l}-q_{l}\right)}_{f_l(p_{2,l},\boldsymbol{q})}}}}, \label{eqn:tr1l}
	\end{align}
	where $f_l(p_{2,l},\boldsymbol{q})$ is given in (\ref{eqn:fl}) on top of the next page.
	\begin{figure*}
	\begin{align}
		\scalebox{0.85}{\mbox{\ensuremath{\displaystyle f_l(p_{2,l},\boldsymbol{q}) = \left[\frac{\frac{1}{\Pi_1}\abs{\rho^{(1)}_{l,l}}^2}{\log{(2)}\left(\sigma^2 + \frac{1}{\Pi_1}\sum_{l'=l}^{M}q_{l'}\abs{\rho^{(1)}_{l,l'}}^2\right)} + \frac{\frac{1}{\Pi_2}\abs{\rho^{(2)}_{l,l}}^2}{\log{(2)}\left(\sigma^2 + \frac{1}{\Pi_2} q_{l}\abs{\rho^{(2)}_{l,l}}^2\right)}\right](p_{2,l}-q_l)}}}\label{eqn:fl}
	\end{align}
	\hrulefill
	\vspace{-0.5cm}
	\end{figure*}
	Lastly, the approximate weighted sum rate, based on (\ref{eqn:tr1l}), is obtained by substituting $\tilde{R}_{1,l}$ for $R_{1,l},l=1,\dots,M,$ in (\ref{eqn:r}) as
	\begin{align}
		\tilde{R}(\boldsymbol{q}) = \scalebox{0.9}{\mbox{\ensuremath{\displaystyle \sum_{l=1}^{M} \mu \tilde{R}_{1,l}(\boldsymbol{q}) + \sum_{l=M+1}^{L} \mu R_{1,l} + \sum_{l=1}^{L} (1-\mu) R_{2,l}.}}} \label{eqn:tr}
	\end{align}
	Exploiting (\ref{eqn:tr}), the weighted sum rate optimization problem can be approximated as
	\begin{align}
	\text{P2$(\boldsymbol{q})$: } \quad \min_{p_{k,l} \geq 0\;\forall\;k,l} -\tilde{R}(\boldsymbol{q}) \qquad \text{s.t.} \quad (\ref{eqn:npp}),
	\end{align}
	which is a convex optimization problem. A suboptimal solution to P1 is found by solving P2$(\boldsymbol{q})$ repeatedly using CCP as described in the following.
	
	Let $\boldsymbol{q}^{(n)}$ denote the value of $\boldsymbol{q}$ in the $n$-th iteration. Starting from $\boldsymbol{q}^{(0)} = \boldsymbol{0},$ in iteration $n=1,2,\dots,$ in accordance with the CCP, the convex problem P2$(\boldsymbol{q}^{(n-1)})$ is solved to obtain the solution $p^{(n)}_{k,l},k=1,2,l=1,\dots,L.$ The obtained solution is used to update $\boldsymbol{q}^{(n)},$ which in the following iteration tightens the linear underestimator in (\ref{eqn:tr1l}). The iterations are continued until convergence (upto a numerical tolerance); the CCP is guaranteed to converge to a stationary point of P1 \cite{Yuille2003}, \cite{Lipp2016}. The proposed algorithm is summarized in Algorithm \ref{alg:pa}.

	\begin{figure}
	\begin{algorithm}[H]
		\begin{algorithmic}[1]
			\STATE {Initialize $\boldsymbol{q}^{(0)} = \boldsymbol{0},$ $p^{(0)}_{k,l} = -\infty,k=1,2,l=1,\dots,L,$ numerical tolerance $\epsilon,$ and iteration index $n=0$}
			\REPEAT
			\STATE{$n \leftarrow n + 1$}
			\STATE {Solve P2($\boldsymbol{q}^{(n-1)}$) and obtain the optimal values $p^{(n)}_{k,l},k=1,2,l=1,\dots,L$ \label{al:p2}}
			\STATE{Update $q^{(n)}_l = p^{(n)}_{2,l}, l=1,\dots,M$ \label{al:q}}
			\UNTIL{$|p^{(n)}_{k,l} - p^{(n-1)}_{k,l}| < \epsilon\;\forall\;k,l$}
			\STATE{Return $p^{(n)}_{k,l},k=1,2,l=1,\dots,L,$ as the solution \label{al:p}}
		\end{algorithmic}
		\caption{Power allocation algorithm for ST MIMO-NOMA.}
		\label{alg:pa}
	\end{algorithm}
	\end{figure}
	\section{Simulation Results}
	\label{sec:sim}
	In this section, we evaluate the performance of the proposed ST MIMO-NOMA scheme using computer simulations. We assume that the first and the second user are located at distances of $d_1 = 250~\text{m}$ and $d_2 = 50~\text{m}$ from the BS, respectively. The path loss is modeled as $\Pi_k = d_k^2,$ i.e., $\Pi_1 = 250^2$ and $\Pi_2 = 50^2.$ Furthermore, we assume $P_T = 30 \text{ dBm}$ and $\sigma^2 = -35 \text{ dBm.}$  
	
	\begin{figure}
		\centering
		\includegraphics[width=0.8\columnwidth]{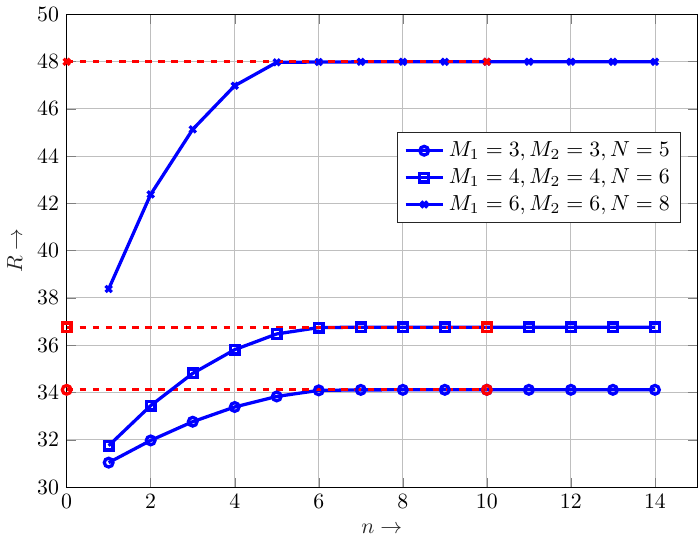}
		\caption{Convergence of Algorithm \ref{alg:pa} for $P_T=30 \text{ dBm.}$}
		\label{fig:conv}
	\end{figure}

	Figure \ref{fig:conv} shows the convergence behavior of Algorithm \ref{alg:pa} for different numbers of transmit and receive antennas and $\mu=0.5.$ From the figure, we observe that the convergence rate of the algorithm is similar for different antenna configurations. In each considered case, the algorithm converges in less than ten iterations. Hence, in the following, we set the number of iterations in Algorithm \ref{alg:pa} to ten.
	
	\begin{figure}
	\centering
	\includegraphics[width=0.8\columnwidth]{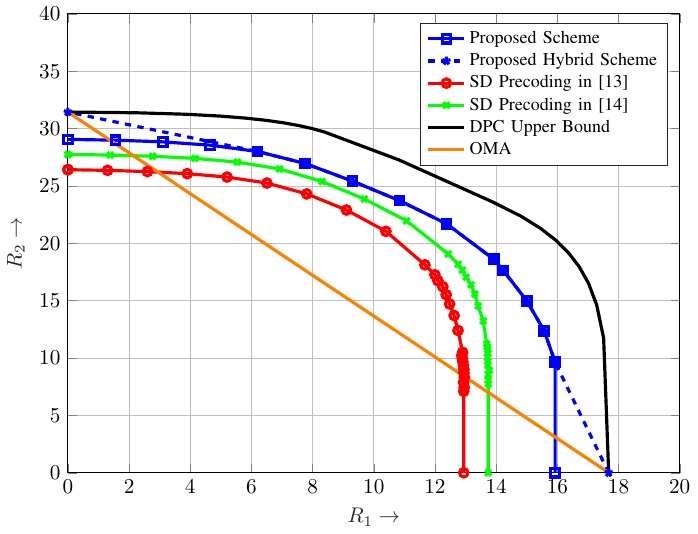}
	\caption{Rate region for $M_1=3, M_2=3, N=5,$ and $P_T = 30 \text{ dBm}.$}
	\label{fig:335}
	\end{figure}	
	
	Figure \ref{fig:335} compares the ergodic rate region of the proposed ST MIMO-NOMA scheme with the ergodic rate regions of DPC, SD precoding in \cite{Chen2019} and \cite{Krishnamoorthy2019}, and OMA. The ergodic rate region of the proposed scheme is computed by varying $\mu$ in (\ref{eqn:r}); the ergodic rate regions of SD precoding are obtained via power allocation analogous to the proposed scheme; the DPC upper bound is obtained by exploiting the broadcast channel (BC)-multiple access channel (MAC) duality, as in \cite{Vishwanath2003}; and the proposed hybrid scheme is obtained by time sharing between the proposed ST MIMO-NOMA scheme and point-to-point MIMO.
	
	From Figure \ref{fig:335}, we observe that the proposed scheme performs close to the DPC upper bound with a gap of about 2 bits per channel use (BPCU). Furthermore, we observe that the proposed scheme outperforms SD precoding and OMA for a wide range of user rates. However, OMA is superior for rate pairs close to point-to-point MIMO. The gap to the DPC upper bound and the performance loss compared to OMA for some user rates are caused by the cancellation of the received signal components corresponding to the off-diagonal elements of the triangularized channel matrix in the proposed scheme. The energy of these signal components cannot be exploited for decoding.
	\section{Conclusion}
	\label{sec:con}
	In this paper, we studied the downlink precoder design for two-user power-domain MIMO-NOMA. First, we presented a method to simultaneously triangularize the MIMO channel matrices of the users using QR decomposition. Next, we utilized ST to design precoding and detection matrices which decompose the MIMO-NOMA channels of the users to SISO-NOMA channels, assuming low-complexity self-interference cancellation at the users. The proposed precoding matrix exploits the null spaces of the MIMO channel matrix of the users, when available, to transmit inter-user-interference free symbols. We also developed a power allocation algorithm based on CCP, which can be used to obtain a lower bound on the ergodic achievable rate region of the proposed ST MIMO-NOMA scheme. The ergodic achievable rate region was compared to those of DPC, SD precoding \cite{Chen2019}, \cite{Krishnamoorthy2019}, and OMA. Computer simulations revealed that the proposed MIMO-NOMA scheme performs close to the DPC upper bound and outperforms SD precoding \cite{Chen2019}, \cite{Krishnamoorthy2019}, and OMA over a wide range of user rates. Further performance improvements were obtained with a hybrid scheme which performs time sharing between the proposed ST MIMO-NOMA scheme and point-to-point MIMO.

	\begin{appendices}
	\renewcommand{\thesection}{\Alph{section}}
	\renewcommand{\thesectiondis}[2]{\Alph{section}:}

	\end{appendices}	
	\bibliographystyle{IEEEtran}
	\bibliography{references}

\begin{thebibliography}{10}
\providecommand{\url}[1]{#1}
\csname url@samestyle\endcsname
\providecommand{\newblock}{\relax}
\providecommand{\bibinfo}[2]{#2}
\providecommand{\BIBentrySTDinterwordspacing}{\spaceskip=0pt\relax}
\providecommand{\BIBentryALTinterwordstretchfactor}{4}
\providecommand{\BIBentryALTinterwordspacing}{\spaceskip=\fontdimen2\font plus
\BIBentryALTinterwordstretchfactor\fontdimen3\font minus
  \fontdimen4\font\relax}
\providecommand{\BIBforeignlanguage}[2]{{%
\expandafter\ifx\csname l@#1\endcsname\relax
\typeout{** WARNING: IEEEtran.bst: No hyphenation pattern has been}%
\typeout{** loaded for the language `#1'. Using the pattern for}%
\typeout{** the default language instead.}%
\else
\language=\csname l@#1\endcsname
\fi
#2}}
\providecommand{\BIBdecl}{\relax}
\BIBdecl

\bibitem{Saito2013}
Y.~Saito, Y.~Kishiyama, A.~Benjebbour, T.~Nakamura, A.~Li, and K.~Higuchi,
  ``Non-orthogonal multiple access {(NOMA)} for cellular future radio access,''
  in \emph{Proc. IEEE 77th Veh. Technol. Conf. (VTC Spring)}, Jun. 2013, pp.
  1--5.

\bibitem{Liu2016}
Y.~Liu, G.~Pan, H.~Zhang, and M.~Song, ``On the capacity comparison between
  {MIMO-NOMA} and {MIMO-OMA},'' \emph{IEEE Access}, vol.~4, pp. 2123--2129, May
  2016.

\bibitem{Zeng2017}
M.~Zeng, A.~Yadav, O.~A. Dobre, G.~I. Tsiropoulos, and H.~V. Poor, ``On the sum
  rate of {MIMO-NOMA} and {MIMO-OMA} systems,'' \emph{IEEE Wireless Commun.
  Lett.}, vol.~6, no.~4, pp. 534--537, Aug. 2017.

\bibitem{Ding2016}
Z.~Ding, R.~Schober, and H.~V. Poor, ``A general {MIMO} framework for {NOMA}
  downlink and uplink transmission based on signal alignment,'' \emph{IEEE
  Trans. Wireless Commun.}, vol.~15, no.~6, pp. 4438--4454, Jun. 2016.

\bibitem{Ali2017}
S.~Ali, E.~Hossain, and D.~I. Kim, ``Non-orthogonal multiple access {(NOMA)}
  for downlink multiuser {MIMO} systems: User clustering, beamforming, and
  power allocation,'' \emph{IEEE Access}, vol.~5, pp. 565--577, Dec. 2017.

\bibitem{Zeng2017a}
M.~Zeng, A.~Yadav, O.~A. Dobre, G.~I. Tsiropoulos, and H.~V. Poor, ``Capacity
  comparison between {MIMO-NOMA} and {MIMO-OMA} with multiple users in a
  cluster,'' \emph{IEEE J. Sel. Areas Commun.}, vol.~35, no.~10, pp.
  2413--2424, Oct. 2017.

\bibitem{Chen2016}
Z.~Chen, Z.~Ding, P.~Xu, and X.~Dai, ``Optimal precoding for a {QoS}
  optimization problem in two-user {MISO-NOMA} downlink,'' \emph{IEEE Commun.
  Lett.}, vol.~20, no.~6, pp. 1263--1266, Jun. 2016.

\bibitem{Chen2017a}
Z.~Chen, Z.~Ding, P.~Xu, X.~Dai, J.~Xu, and D.~W.~K. Ng, ``Comment on
  ``Optimal precoding for a {QoS} optimization problem in two-user {MISO-NOMA}
  downlink'','' \emph{IEEE Commun. Lett.}, vol.~21, no.~9, pp. 2109--2111,
  Sep. 2017.

\bibitem{Ding2016b}
Z.~Ding, L.~Dai, and H.~V. Poor, ``{MIMO-NOMA} design for small packet
  transmission in the {Internet of Things},'' \emph{IEEE Access}, vol.~4, pp.
  1393--1405, Apr. 2016.

\bibitem{Choi2016}
J.~Choi, ``On the power allocation for {MIMO-NOMA} systems with layered
  transmissions,'' \emph{IEEE Trans. Wireless Commun.}, vol.~15, no.~5, pp.
  3226--3237, May 2016.

\bibitem{Wang2019}
B.~{Wang}, R.~{Shi}, C.~{Ji}, and J.~{Hu}, ``Joint precoding and user
  scheduling for full-duplex cooperative {MIMO-NOMA} {V2X} networks,'' in
  \emph{Proc. IEEE 90th Vehicular Technology Conf.}, Sep. 2019, pp. 1--6.

\bibitem{Xiao2019}
Y.~S. {Xiao} and D.~H.~K. {Tsang}, ``Interference alignment beamforming and
  power allocation for cognitive {MIMO-NOMA} downlink networks,'' in
  \emph{Proc. IEEE Wireless Commun. and Netw. Conf.}, Apr. 2019, pp. 1--6.

\bibitem{Chen2019}
Z.~Chen, Z.~Ding, X.~Dai, and R.~Schober, ``Asymptotic performance analysis of
  {GSVD-NOMA} systems with a large-scale antenna array,'' \emph{IEEE Trans.
  Wireless Commun.}, vol.~18, no.~1, pp. 575--590, Jan. 2019.

\bibitem{Krishnamoorthy2019}
A.~{Krishnamoorthy}, R.~{Schober}, and Z.~{Ding}, ``Downlink precoder design
  for two-user power-domain {MIMO-NOMA} with excess degrees of freedom,'' in
  \emph{Proc. IEEE Int. Conf. Commun.}, May 2019, pp. 1--6.

\bibitem{Krishnamoorthy2019a}
A.~{Krishnamoorthy} and R.~{Schober}, ``Precoder design for two-user uplink
  {MIMO-NOMA} with simultaneous triangularization,'' in \emph{Proc. IEEE Global
  Commun. Conf.}, Dec. 2019, pp. 1--6.

\bibitem{Hanif2019}
M.~F. Hanif and Z.~Ding, ``Robust power allocation in {MIMO-NOMA} systems,''
  \emph{IEEE Wireless Commun. Lett.}, vol.~8, pp. 1541--1545, Dec. 2019.

\bibitem{Yuille2003}
A.~L. Yuille and A.~Rangarajan, ``The concave-convex procedure,'' \emph{Neural
  {C}omputation}, vol.~15, no.~4, pp. 915--936, 2003.

\bibitem{Lipp2016}
T.~Lipp and S.~Boyd, ``Variations and extension of the convex--concave
  procedure,'' \emph{Optimization and Engineering}, vol.~17, no.~2, pp.
  263--287, 2016.

\bibitem{Vishwanath2003}
S.~{Vishwanath}, N.~{Jindal}, and A.~{Goldsmith}, ``Duality, achievable rates,
  and sum-rate capacity of gaussian {MIMO} broadcast channels,'' \emph{IEEE
  Trans. Inf. Theory}, vol.~49, no.~10, pp. 2658--2668, Oct. 2003.

\bibitem{Golub2012}
G.~H. Golub and C.~F. Van~Loan, \emph{Matrix {C}omputations}.\hskip 1em plus
  0.5em minus 0.4em\relax JHU Press, 2012, vol.~3.

\bibitem{WangGiannakis2011}
X.~Wang and G.~B. Giannakis, ``Resource allocation for wireless multiuser
  {OFDM} networks,'' \emph{IEEE Trans. Inf. Theory}, vol.~57, no.~7, pp.
  4359--4372, 2011.

\bibitem{Boyd2004}
S.~Boyd, S.~P. Boyd, and L.~Vandenberghe, \emph{Convex {O}ptimization}.\hskip
  1em plus 0.5em minus 0.4em\relax Cambridge {U}niversity {P}ress, 2004.

\end{thebibliography}
	
\end{document}